\documentclass[journal,doublecolumn,10pt]{IEEEtran}
\usepackage{epsfig,latexsym}
\usepackage{float}
 \usepackage{enumerate}
\usepackage{indentfirst}
\usepackage{amsmath}
\usepackage{amssymb}
\usepackage{times}
\usepackage{subfigure}
\usepackage{cite}
\usepackage{url}
\usepackage{color}
\usepackage{graphicx}
\usepackage{epstopdf}
\definecolor{MyDarkBlue}{rgb}{0,0.08,0.45}
\definecolor{yellow}{rgb}{0.99,0.99,0.70}
\definecolor{white}{rgb}{0.796,0.948,0.816}
\definecolor{black}{rgb}{0.00,0.00,0.00}

\def\diag{\textrm{diag}}

\newtheorem{Lemma}{$\mathbf{Lemma}$}
\newtheorem{Corollary}{$\mathbf{Corollary}$}

\newtheorem{Remark}{Remark}

\newtheorem{theorem}{$\mathbf{Theorem}$}

\newtheorem{proof}{Proof}

\begin{document}
\title{\huge Ergodic Secrecy Rate of RIS-Assisted Communication Systems   in the Presence of Discrete Phase Shifts \\and Multiple Eavesdroppers}

\author{Peng Xu,  \IEEEmembership{Member, IEEE},
Gaojie Chen,  \IEEEmembership{Senior Member, IEEE},
Gaofeng Pan,  \IEEEmembership{Senior Member, IEEE}, \\and Marco Di Renzo, \IEEEmembership{Fellow, IEEE}

\thanks{
P. Xu is  with  Chongqing Key Laboratory of Mobile Communications Technology,  School of   Communication and Information Engineering,  Chongqing University of Posts and Telecommunications, Chongqing, 400065,  China.
(e-mail: xupeng@cqupt.edu.cn).

G. Chen is with School of Engineering, University of Leicester, Leicester LE1 7RH, U.K. (e-mail: gaojie.chen@leicester.ac.uk).

G. Pan is with the School of Information and Electronics Engineering, Beijing Institute of Technology, Beijing 100081, China, and he is also with Computer, Electrical and Mathematical Sciences and Engineering Division, King Abdullah University of Science and Technology (KAUST), Thuwal 23955-6900, Saudi Arabia.  (e-mail: gaofeng.pan.cn@ieee.org).

M. Di Renzo is with Universit\'e Paris-Saclay, CNRS and CentraleSup\'elec, Laboratoire des Signaux et Syst\`emes,  Gif-sur-Yvette, France. (e-mail: marco.direnzo@centralesupelec.fr).
}\vspace{-2em}}
\maketitle

\begin{abstract}
This letter investigates the ergodic secrecy rate (ESR) of a reconfigurable intelligent surface (RIS)-assisted communication system  in the presence of discrete phase shifts and multiple eavesdroppers (Eves).
In particular, a closed-form approximation  of the ESR is derived  for both non-colluding and colluding Eves. The analytical results are shown to be accurate when the number of reflecting elements of the RIS $N$ is large.
{Asymptotic analysis is  provided to
investigate the impact of $N$  on the ESR, and it is proved that the
ESR scales with $\log_2 N$  for
both non-colluding and colluding Eves.}
Numerical results are provided to verify the analytical results and the obtained scaling laws.
\vspace{-0.5em}\end{abstract}


\vspace{-1em}
\section{Introduction}
Reconfigurable intelligent surfaces (RISs) utilize a large number of passive reflecting elements to customize wireless communication environments \cite{IRS_Marco_Eurasip,IRS_Marco_Smart,Marco2019access,IRS_Marco_Relay}.
 Due to the low-cost, high energy-efficiency and full-duplex advantages, RISs are regarded as a promising technology for next-generation wireless communications and hence have recently received significant
  academic and industrial attention \cite{wu2019intell,Liu2020RIS_Surveys,pan2020multicell,pan2020intell}.


RISs have various potential applications in wireless communications, which include the design of secure wireless systems based on the concept of physical layer security
(e.g.,
\cite{IRS_Security_Cui,IRS_Security_Chu,IRS_Security_AN,
IRS_Security_Jasc,wang2020intell_spl,yang2020vt}). In
\cite{IRS_Security_Cui,IRS_Security_Chu,IRS_Security_AN,IRS_Security_Jasc},
the authors
investigated optimization problems to jointly design the beamforming vectors and phase shifts at
the transmitter and RIS, respectively. In general, there exist  two objectives  for the design of the phase
shifts at the RIS: (i) to strengthen the legitimate channels by co-phasing the
reflected signals with the signal directly received from the transmitter; and   (ii) to
suppress the eavesdropping channels by setting the reflected signals at the eavesdroppers
(Eves) to be in opposite phase with respect to the signal from the transmitter. The key idea behind  the
optimization problems in existing works \cite{IRS_Security_Cui,IRS_Security_Chu,IRS_Security_AN,IRS_Security_Jasc}  lies in  achieving a favorable trade-off
between these two design objectives, which requires the knowledge of the instantaneous eavesdropping
channel state information (CSI) at the transmitter and RIS. However, the instantaneous
eavesdropping CSI is difficult  to obtain in practice, since the Eves are usually
passive and do not actively communicate with other nodes. Motivated by this consideration,
 the authors  of \cite{wang2020intell_spl} and \cite{yang2020vt} considered  RIS-assisted secrecy
  communications without assuming the knowledge of the instantaneous eavesdropping CSI.


Different from these existing works,
this letter investigates the ergodic secrecy rate (ESR) of RIS-assisted systems
  in the presence of discrete phase shifts and multiple Eves. 
In particular, by approximately characterizing the  distribution
of  the  received  signal-to-noise-ratios  (SNRs)  at  the Eves,
 we obtain a closed-form approximation of the ESR for both non-colluding and colluding Eves.
 The analysis of the ESR, in fact,  is essentially  different from the analysis of  the ergodic rate
  without security constraints \cite{IRS_Error_Capacity}, since the phase shifts at the RIS lead to a different impact on the intended receiver and Eves.
  Moreover, based on passive beamforming,  the received SNRs at the destination and
  Eves  depend on  phase quantization  errors and cascaded channels, that  are
different from those in  massive multiple-input multiple-output (MIMO) systems.
 In order to provide insights, asymptotic analysis is also provided, which shows that
 the ESR scales with $\log N$  for both non-colluding and colluding Eves.
 Numerical results are illustrated  to verify that the analytical results  are accurate for large values of $N$.

\textit{Notation}: $\mathbb{C}$ and $\mathbb{Z}$ denote the  complex domain and integer set, respectively;   we denote $[1:M]\triangleq\{1,\ldots,M\}$, where $M$ is a positive integer,  and $[x]^+\triangleq\max\{0,x\}$;  $\mathcal{CN}$ denotes the complex Gaussian distributions; $\mathbb{E}[\cdot]$ denotes the  expectation of a random variable; $\log(\cdot)$ and $\ln(\cdot)$ denote the base-two  and natural logarithms, respectively; and $\kappa$ is Euler's constant.



\vspace{-1.5em}
\section{System Model and Preliminaries}\label{section_model}
We consider an RIS-assisted secure communication system  with a source ($S$), an RIS ($R$) with $N$ reconfigurable  elements, a destination ($D$) and $K$ Eves ($E_k$, $\forall k\in[1:K]$).
The reconfigurable elements of the RIS are arranged in a uniform array of tiny antennas spaced half of the wavelength apart.
All nodes are assumed to be equipped with a single
antenna\footnote{RIS-aided transmission has several  applications when multiple antennas are  not
 available at either the transmitter or  the receiver, e.g., in device-to-device communications
 \cite{IRS_Marco_Smart}.}. The channels $S\rightarrow D$,  $S\rightarrow E_k$, $S\rightarrow R$, $R\rightarrow D$ and $R\rightarrow E_k$ are denoted by $h_{SD}\in\mathbb{C}$, ${h}_{SE_k}\in\mathbb{C}$,
$\mathbf{h}_{SR}\in\mathbb{C}^{N\times1}$,  $\mathbf{h}_{RD}\in\mathbb{C}^{N\times1}$
and $\mathbf{h}_{k}\in\mathbb{C}^{N\times1}$, respectively.
These channels are modeled as $h_{SD}=g_{SD}d_{SD}^{-\frac\alpha2}$,  ${h}_{SE_k}=g_{SE_k}d_{SE_k}^{-\frac\alpha2}$,
$[\mathbf{h}_{SR}]_n=g_{SR,n}d_{SR}^{-\frac\alpha2}$,  $[\mathbf{h}_{RD}]_n=g_{RD,n}d_{RD}^{-\frac\alpha2}$
and $[\mathbf{h}_{k}]_n=g_{k,n}d_{k}^{-\frac\alpha2}$, where $g_{SD},~g_{SE_k},~g_{SR,n},~g_{RD,n},~
g_{k,n}\sim \mathcal{CN}(0,1)$ denote the small-scale
fading\footnote{
Since the elements of the RIS are spaced half of the wavelength apart and we assume that the location of the RIS cannot be optimized to ensure  strong line-of-sight links, we have, as a first approximation similar to \cite{yang2020vt,IRS_Error_Capacity,IRS_Marco}, that the channels can be modeled as independent and identically distributed, and follow a Rayleigh distribution.
},
$d_{SD}$, $d_{SE_k}$, $d_{SR}$,
  $d_{RD}$ and $d_{k}$ denote the distances  $S\rightarrow D$,  $S\rightarrow E_k$, $S\rightarrow R$,
$R\rightarrow D$ and $R\rightarrow E_k$, respectively
, $\forall n\in[1:N],~k\in[1:K]$, and $\alpha$ is the pass-loss exponent.  Then,
 the received signal at $D$ and $E_k$ can be written as

\vspace{-1em}
{\small
\begin{align}
  y_D&=\sqrt{P}\big(\eta\mathbf{h}_{SR}^T\boldsymbol{\Phi}\mathbf{h}_{RD}+h_{SD}\big)x_S+n_D,\label{y_d}\\
  y_{E_k}&=\sqrt{P}\big(\eta\mathbf{h}_{SR}^T\boldsymbol{\Phi}\mathbf{h}_{k}+h_{SE_k}\big)x_S+n_{E_k},\label{y_ek}
\end{align}}
respectively, where $x_S$ is the transmitted signal, $\mathbb{E}(|x_S|^2)=1$,  $P$ is the transmit power,
$n_D$ and $n_{E_k}\sim\mathcal{CN}(0,\delta^2)$ 
are the additive white
Gaussian noises at $D$ and $E_k$, respectively,
$\eta\in(0,1]$ is the
amplitude reflection coefficient,  $\boldsymbol{\Phi}\triangleq \diag(e^{j\phi_1},\ldots,e^{j\phi_N})$
and $\phi_n\in[0,2\pi)$ is the phase shift of the $n$th element of  the RIS.

{We assume that the RIS does not have access to  the instantaneous  eavesdropping CSI, so that
it cannot design $\phi_n$  in order to suppress the received SNRs at the Eves.
However,  the RIS is assumed to know the  instantaneous legitimate CSI.
 Under these assumptions,  the optimal value  of $\phi_n$
  that maximizes the received SNR at $D$ is {\small{$\phi_n^*=\theta_{SD}\!-\theta_{SR,n}\!-\theta_{RD,n}$}}, where
$\theta_{SD}$, $\theta_{SR,n}$ and $\theta_{RD,n}$ denote the phases of $g_{SD}$, $g_{SR,n}$ and $g_{RD,n}$,
respectively.}
Due to hardware limitations,  $\phi_n$ can only take a finite number of discrete values.
In particular, the set of discrete phase shifts is denoted by {\small{$\mathcal{F}\!\triangleq\!\left\{0,\frac{2\pi}{2^b},\ldots,\frac{(2^b-1)2\pi}{2^b}\right\}$}},
where $b$ denotes the  number of quantization bits.
Accordingly, we set $\phi_n\!=\!f_1(\phi_n^*)$, where the function $f_1(\phi_n^*)$
maps $\phi_n^*$ to the nearest point in $\mathcal{F}$, i.e., 
\begin{equation}\small\label{f_1}
  f_1(\phi_n^*)=\hat{\phi_i},\textrm{ if }|\phi_n^*-\hat{\phi_i}|\leq |\phi_n^*-\hat{\phi_j}|,\
\hat{\phi_i},\hat{\phi_j}\in\mathcal{F},\  \forall j\neq i.
\end{equation}
Therefore, the phase quantization error is $\Theta_n=
f_1(\phi_n^*)-\phi_n^*$, which is uniformly distributed in $\left[-\frac{\pi}{2^b},\frac{\pi}{2^b}\right]$,
 similar to \cite{IRS_Error_Capacity,IRS_diversity}.
According to \eqref{y_d} and \eqref{y_ek},
the received SNRs at $D$ and $E_k$ can be formulated, respectively, as follows

\vspace{-1em}
{\small{\begin{align}\label{gamma_D}
  \gamma_D&={\rho\left||h_{SD}|+\eta\sum_{n=1}^N\left|[\mathbf{h}_{SR}]_n[\mathbf{h}_{RD}]_n\right|e^{j\Theta_n}
  \right|^2}\notag\\
  &=\rho\left|d_{SD}^{-\frac\alpha2}|g_{SD}|+\eta d_{SR}^{-\frac\alpha2}d_{RD}^{-\frac\alpha2}
  \sum_{n=1}^N|g_{SR,n}g_{RD,n}|e^{j\Theta_n}\right|^2,\\
  \gamma_{E_k}&=\rho\left|h_{SE_k}+\eta\sum_{n=1}^N\left|[\mathbf{h}_{SR}]_n[\mathbf{h}_{k}]_n\right|
e^{j\psi_{k,n}}\right|^2\notag\\
 & =\rho\left|d_{SE_k}^{-\frac\alpha2}g_{SE_k}+\eta d_{SR}^{-\frac\alpha2} d_{k}^{-\frac\alpha2}\sum_{n=1}^N\left|g_{SR,n}g_{k,n}\right|
e^{j\psi_{k,n}}\right|^2,\label{gamma_Ek}
\end{align}}}
\hspace{-0.3em}where 
$\psi_{k,n}\triangleq f_2(\phi_n^*,\theta_{SR,n})+\theta_{k,n}$, $\theta_{k,n}$ is the phase of $g_{k,n}$
and the function $f_2(\phi_n^*,\theta_{SR,n})$ is defined as follows
\begin{equation}\small\label{f_2}
f_2(\phi_n^*,\theta_{SR,n})\triangleq
f_1(\phi_n^*)+\theta_{SR,n}.
\end{equation}
The ESR\footnote{We assume that the RIS appropriately customizes the wireless channel but we consider that the distribution of the signal transmitted by $S$ is always Gaussian. It is worth mentioning that the information-theoretic characterization of RIS-assisted transmission and the calculation of the optimal input distribution in the presence of an RIS is an open issue that is currently under active research \cite{Marco_ISIT}. This is, however, beyond the scope of this letter.}
 can be expressed as follows 
\begin{equation}\small\label{R_s}
  R_s=[R_D-R_E]^+,
\end{equation}
where $R_D=\mathbb{E}_{\gamma_D}[\log(1+\gamma_D)]$ and  $R_E$ denote the ergodic rates  from $S$ to $D$ and
from $S$ to the Eves, respectively.
Given $\{\Theta_n\}_{n=1}^{N}$, an approximated  expression of $R_D$ can be found in \cite[Eq. (13)]{IRS_Error_Capacity}. By averaging  over $\{\Theta\}_{n=1}^N$, $R_D$ can be calculated as shown in \eqref{R_D} at the top of the next page,  where $A_1\triangleq \eta^2d_{SR}^{-\alpha}d_{RD}^{-\alpha}$, $A_2\triangleq\frac{\sqrt{\pi}\eta2^b}4d_{SD}^{-\frac\alpha2}d_{SR}^{-\frac\alpha2}d_{RD}^{-\frac\alpha2}
 \sin\frac\pi {2^b}$ and $A_3\triangleq \frac{\eta^22^{2b}}{32}d_{SR}^{-\alpha}d_{RD}^{-\alpha}
  \left(1\!-\!\cos\frac{2\pi}{2^b}\right)$.

\begin{figure*}
{\small{
\begin{align}\label{R_D}
  R_D&\!\approx \!\log\left(\!1\!+\!\rho\bigg( \!N\eta^2d_{SR}^{-\alpha}d_{RD}^{-\alpha}\!+\!d_{SD}^{-\alpha}
  \!+\!\frac{\pi^{\frac32}\eta}4d_{SD}^{-\frac\alpha2}d_{SR}^{-\frac\alpha2}d_{RD}^{-\frac\alpha2}
  \sum_{n=1}^N\mathbb{E}_{\Theta_n}[\cos\Theta_n] \!+\!\frac{\pi^{2}\eta^2}8d_{SR}^{-\alpha}d_{RD}^{-\alpha}\sum_{i=1}^{N\!-1}
  \!\sum_{k=i\!+\!1}^N\!\mathbb{E}_{\Theta_i,\Theta_k}[\cos(\Theta_k\!-\!\Theta_i)]\bigg)\!\right)\notag\\
 & =\log\left(1+\rho NA_1+\rho d_{SD}^{-\alpha}
  +\rho NA_2+\rho N(N-1)A_3\right),
\end{align}}}
\vspace{-1em}
\hrule
\end{figure*}

\begin{Remark}
  The analysis of the ESR for the considered RIS-assisted system relies only  on the
  knowledge of the statistical eavesdropping CSI, which can be obtained by using several methods, such as those used in \cite{stati_CSI2}.
\end{Remark}

In the following sections, $R_E$  is calculated for non-colluding and colluding Eves, respectively.
\vspace{-1em}
\section{Non-Colluding Eves}\label{section_capacity_nonC}
In the non-colluding case, $R_E$ can be expressed as follows
\begin{equation}\small\label{R_ek}
  R_{E}=\max_{k\in[1:K]}R_{E_k},
  \end{equation}
  where  $R_{E_k}\triangleq\mathbb{E}_{\gamma_{E_k}}[\log(1+\gamma_{E_k})]$.
  In order to derive $R_{E_k}$, the distribution of $\gamma_{E_k}$ in \eqref{gamma_Ek} needs to be computed.
  \vspace{-1em}
  \subsection{Distribution of $\gamma_{E_k}$}\label{subsection_gamma_ek}
  Before deriving the distribution of $\gamma_{E_k}$, we introduce the following lemma.
     \begin{Lemma}\label{lemma_phase}
The phase $\psi_{k,n}$, $k\in[1:K]$, $n\in[1:N]$,  in \eqref{gamma_Ek} has the following properties:
  \begin{enumerate}
    \item[a)] $\psi_{k,n}$ is uniformly distributed in $[0,2\pi)$;
   \item[b)] $\psi_{k,n}$ is independent of $f_2(\phi_n^*,\theta_{SR,n})$ defined in \eqref{f_2};
    \item[c)] $\psi_{k,i}$ is  independent of $\psi_{k,j}$, $\forall i\neq j,~i,j\in[1:N]$.
  \end{enumerate}
  \end{Lemma}
\begin{proof}
See Appendix \ref{proof_lemma_phase}.
\end{proof}


Based on Lemma \ref{lemma_phase}, the distribution of $\gamma_{E_k}$ in \eqref{gamma_Ek}
is provided in the following lemma.
\begin{Lemma}\label{lemma_gamma_ek}
When $N$ is large,   $\gamma_{E_k}$  can be approximated with
an exponential random  variable with  mean $\lambda_{E_k}=\rho\big(d_{SE_k}^{-\alpha}+NB_k
\big)$, where $B_k\triangleq \eta^2
d_{SR}^{-\alpha} d_{k}^{-\alpha}$. 
\end{Lemma}
\begin{proof}
 Define $G_k\triangleq\sum_{n=1}^N\left|g_{SR,n}g_{k,n}\right|
e^{j\psi_{k,n}}$, $k\in[1:K]$. Based on \cite[Lemma 2]{IRS_NOMA_Ding} and the fact that
 $\{\psi_{k,n}\}_{n=1}^N$ are  independent and identically distributed
uniform random variables in $[0,2\pi)$ as proved in Lemma \ref{lemma_phase},
$G_k\sim \mathcal{CN}(0,N)$  as $N\rightarrow\infty$.
Furthermore, since $g_{SE_k}$ is independent of $G_{k}$, we have
 \begin{equation}\small{
 d_{SE_k}^{-\frac\alpha2}g_{SE_k}\!+\!\eta d_{SR}^{-\frac\alpha2} d_{k}^{-\frac\alpha2}G_{k}
 \sim \mathcal{CN}\left(0,d_{SE_k}^{-\alpha}\!+\!NB_k\right)},\notag\end{equation}
  as $N\rightarrow\infty$. Recalling that
$\small \gamma_{E_k}=\rho\left|d_{SE_k}^{-\frac\alpha2}g_{SE_k}+\sqrt{B_k}G_{k}
\right|^2$  in \eqref{gamma_Ek}, the proof follows.
\end{proof}
\begin{Remark}
 The authors of  \cite{IRS_Error_Capacity} approximated $R_D$ based on the fact  that
 Jensen's inequality is tight, rather than
an upper   bound, 
if $\frac{{\rm Var}[\gamma_D]}{\mathbb{E}^2[\gamma_D]}\to 0$, as $N\to \infty$. However, Lemma \ref{lemma_gamma_ek} shows that
 Jensen's inequality cannot be applied for approximating  
$R_E$, since $\frac{{\rm Var}[\gamma_{E_k}]}{\mathbb{E}^2[\gamma_{E_k}]}\to 1$, as $N\to \infty$.
\end{Remark}
  \vspace{-1em}
\subsection{Ergodic Secrecy Rate}
The ESR for non-colluding Eves is summarized in the following theorem.

  \begin{theorem}\label{theorem_non}
    When $N$ is large, the ESR for  non-colluding Eves can be expressed as follows
\begin{equation}\small\label{R_s_n}
R_s\thickapprox\left[R_D+\frac{1}{\ln2}\max_{k\in[1:K]}{e^{\frac1{\lambda_{E_k}}}} {\rm Ei}\left(-\frac1{\lambda_{E_k}}\right)\right]^+,
  \end{equation}
where   $\rm Ei(x)\triangleq -\int_{-x}^\infty\frac{e^{-t}}t{\rm d}t,$ $x<0$, is the exponential integral function \cite[Eq. 8.211]{gradsh2000table}.
  \end{theorem}

  \begin{proof}
  When $N$ is  large, based on Lemma \ref{lemma_gamma_ek},  $R_{E_k}$ in \eqref{R_ek} can be approximated  as follows
  \begin{align}\label{R_ek2}
  R_{E_k}&\approx\int_{0}^{\infty}\log(1+x)\frac1{\lambda_{E_k}}e^{-\frac x{\lambda_{E_k}}}{\rm d}x
  =\frac{1}{\ln2}\int_{0}^{\infty}\frac{e^{-\frac x{\lambda_{E_k}}}}{1+x}{\rm d}x\notag\\
  &{=}-\frac{e^{\frac1{\lambda_{E_k}}}}{\ln2} {\rm Ei}\left(-\frac1{\lambda_{E_k}}\right),
  \end{align}
  where the last equality   is based on \cite[Eq. 3.352.4]{gradsh2000table}.
  Combining \eqref{R_s}, \eqref{R_ek} and \eqref{R_ek2},
 the theorem is proved.
  \end{proof}
\vspace{-1em}
{\subsection{Asymptotic Analysis}
  To obtain insights from the obtained ESR,  its
asymptotic behavior is analyzed  in the following corollary.
\begin{Corollary}\label{corollary_scale_N}
 \textrm{As } $N\rightarrow\infty$, $R_s\rightarrow \log N+\log A_3+\frac{\kappa}{\ln2}-\max_{k\in[1:K]}\log B_k$,
 which implies that the ESR for non-colluding Eves scales with $\log N$.
\end{Corollary}
\begin{proof}
  From \eqref{R_D}, we have
  \begin{align}\label{R_D_a}
    \! R_D&\!\approx\! \log\left(\! \rho A_3N^2 \left(\! \frac1{\rho A_3N^2}\!+\!\frac{A_1}{A_3N}\!+\!\frac{d_{SD}^{-\alpha}}{A_3N^2}\!+\!\frac{A_2\!-\!A_3}{A_3 N}\!+\!1\! \right) \right)\notag\\
    &\rightarrow 2\log N+\log\rho+\log A_3, \textrm{ as }N\rightarrow\infty.
  \end{align}
  In addition, $R_{E_k}$ in \eqref{R_ek2} can be further expressed as follows
    \begin{align}\label{R_ek3}
  R_{E_k}&\stackrel{(a)}{\approx} \frac{e^{\frac1{\lambda_{E_k}}}}{\ln2}\left(-\kappa+\ln({\lambda_{E_k}})+\sum_{i=1}^{\infty}
  \frac{(-1)^{i+1}}{i\cdot i!\cdot\lambda_{E_k}^i}\right)\notag\\
  &\stackrel{(b)}{\rightarrow}\log({\lambda_{E_k}})-\frac{\kappa}{\ln2}\notag\\
   &\stackrel{(c)}{=}\log\left(N\rho B_k\right)+\log\left(1+\frac{d_{SE_k}^{-\alpha}}{NB_k}\right)-\frac{\kappa}{\ln2}
   \notag\\
    &\rightarrow \log N+\log \rho+\log B_k-\frac{\kappa}{\ln2}, \textrm{ as }N\rightarrow\infty.
    \end{align}
     where $(a)$  is based on \eqref{R_ek2} and \cite[Eq. 8.214.1]{gradsh2000table},  $(b)$ holds since $1/\lambda_{E_k}\rightarrow 0$ as $N\rightarrow\infty$, and $(c)$ is based on the definition of $\lambda_{E_k}$ in Lemma \ref{lemma_gamma_ek}.

     Combining \eqref{R_s}, \eqref{R_ek}, \eqref{R_D_a} and \eqref{R_ek3}, the corollary follows.
\end{proof}
}

  \begin{Remark}
   { Compared with the scaling law  $2\log N$ for non-secrecy transmission with discrete phase shifts
    \cite{IRS_Error_Capacity},
    Corollary \ref{corollary_scale_N}  shows that
    the ESR scales with  $\log N$.}
    \end{Remark}

\vspace{-1em}

\section{Colluding Eves}
When the  Eves are  colluding, they can combine their received signals for information interception.
 Based on \cite{Collud_ISIT},  $R_E$ in \eqref{R_s} 
can be expressed as follows
\begin{equation}\small\label{R_E_cd}
  R_{E}=\mathbb{E}_{\{\gamma_{E_k}\}_{k=1}^K}\log\left(1+\sum_{k=1}^K\gamma_{E_k}\right).
\end{equation}
Since
common random variables $\{g_{SR,n}\}_{n=1}^N$ are present  in every $\gamma_{E_k}$ as shown in \eqref{gamma_Ek}, $\{\gamma_{E_k}\}_{k=1}^K$ are correlated random variables.
However, the following lemma shows that such correlation is negligible for large  values of $N$.

   \begin{Lemma}\label{lemma_indep}
     $\gamma_{E_i}$ is independent of $\gamma_{E_j}$ if $N\rightarrow \infty$, $i,j\in[1:K],~ i\neq j$.
   \end{Lemma}
   \begin{proof}
   Let us define {\small{$ H_{E_k}\triangleq d_{SE_k}^{-\frac\alpha2}g_{SE_k}+\sqrt{B_k}\sum_{n=1}^N\left|g_{SR,n} g_{k,n}\right|
e^{j\psi_{k,n}}$}}. Thus,  $\gamma_{E_k}=\rho|H_{E_k}|^2$. According to Lemma \ref{lemma_phase}-a),
$\mathbb{E}\left[e^{j\psi_{k,n}}\right]=0$ and {\small{$\mathbb{E}\left[H_{E_k}\right]=0$}},  $\forall k\in[1:K].$
Moreover,  since $\psi_{i,n}$ is independent of  $\psi_{j,m}$ if $i\neq j$ or $n\neq m$,
 we have {\small{$\mathbb{E}[H_{E_i}H_{E_j}]=0$,  $\forall i\neq j.$}}
Therefore, the covariance of $H_{E_i}$ and $H_{E_j}$ is zero.
Furthermore, based on Lemma \ref{lemma_gamma_ek}, $\{H_k\}_{k=1}^K$
are uncorrelated complex Gaussian variables if $N\rightarrow \infty$, and hence $\{H_k\}_{k=1}^K$ are independent of each other.
This completes the proof.
   \end{proof}

   \vspace{-1em}
   \subsection{Ergodic Secrecy Rate}

Based on  Lemma \ref{lemma_gamma_ek} and Lemma \ref{lemma_indep}, the ESR  is provided in the following theorem.
   \begin{theorem}\label{theorem_cd}
 When $N$ is large,
 and $\lambda_{E_i}\neq \lambda_{E_j}$, $\forall i\neq j$, $i,j\in[1:N]$, the ESR for   colluding Eves
  can be approximated  as follows
  \begin{equation}\small\label{R_s_c}
R_s\!\approx\!\! \left[R_D\!+\!\frac{1}{\ln2}\sum_{i=1}^Ke^{\frac{1}{\lambda_{E_i}}}{\rm Ei}\left(-\frac{1}{\lambda_{E_i}}\right)\!\prod_{j=1,j\neq i}^K\frac{\lambda_{E_i}}{\lambda_{E_i}\!-\!{\lambda_{E_j}}}
\!\right]^+\!\!.
  \end{equation}
  \end{theorem}
    \begin{proof}
     Based on Lemma \ref{lemma_gamma_ek} and Lemma \ref{lemma_indep}, if $\lambda_{E_i}\neq \lambda_{E_j}$, $\forall i\neq j$,
      $\sum_{k=1}^K\gamma_{E_k}$ has the following probability density function (PDF)
      \cite{ross2010proba}:
    \begin{equation}\small\label{pdf_distri}
      f_{\sum_{k=1}^K\gamma_{E_k}}(x)\!=\!\sum_{i=1}^K\frac{1}{\lambda_{E_i}}e^{-\frac{x}{\lambda_{E_i}}}
      \prod_{j=1,j\neq i}^K\frac{\lambda_{E_i}}{\lambda_{E_i}\!-\!{\lambda_{E_j}}}.
    \end{equation}

 Combining  \eqref{R_ek2}, \eqref{R_E_cd} and \eqref{pdf_distri}, we obtain 
    \begin{equation}\small\label{R_E_cd2}
      R_E\thickapprox-\frac1{\ln2}\sum_{i=1}^K{e^{\frac1{\lambda_{E_i}}}{\rm Ei}\left(-\frac{1}{\lambda_{E_i}}\right)} \prod_{j=1,j\neq i}^K\frac{\lambda_{E_i}}{\lambda_{E_i}-{\lambda_{E_j}}}
      .
    \end{equation}
  Recalling \eqref{R_s}, the theorem is proved. 
     \end{proof}

 \begin{Remark}
 {Theorem \ref{theorem_cd} corresponds to the case that the Eves lie in different locations,
so that  $\{\gamma_{E_k}\}_{k=1}^K$ have different means.
When the  Eves are clustered relatively
closely together,  $\{\gamma_{E_k}\}_{k=1}^K$ have the same (or a very similar)  mean. In this case,
  the ESR
 can be analyzed in a similar way, whose details are not provided due to space limitations.}
  \end{Remark}
\vspace{-1em}
  {\subsection{Asymptotic Analysis}
 The
asymptotic behavior of the obtained ESR for colluding Eves is provided in the following corollary.
\begin{Corollary}\label{corollary_scale_N2}
   As $N\rightarrow\infty$, $R_s\rightarrow \log N+\log A_3+\frac{\kappa}{\ln2}-\sum_{i=1}^K\log B_i\prod_{j=1,j\neq i}^K\frac{B_i}{B_i-B_j}$, which implies that the ESR for colluding Eves also scales with $\log N$.
\end{Corollary}
\begin{proof}
  From \eqref{R_ek2}, \eqref{R_ek3} and \eqref{R_E_cd2}, we have
  \begin{align}\label{R_E_cd3}
    R_E&\rightarrow \!\sum_{i=1}^K \left(\log N\!+\!\log \rho\!-\!\frac{\kappa}{\ln2}\!+\!\log B_i \right) \!\prod_{j=1,j\neq i}^K\frac{\lambda_{E_i}}{\lambda_{E_i}\!-\!{\lambda_{E_j}}}\notag\\
    &\stackrel{(a)}\rightarrow \sum_{i=1}^K \left(\log N\!+\!\log \rho\!-\!\frac{\kappa}{\ln2}\!+\!\log B_i \right) \prod_{j=1,j\neq i}^K\frac{B_{i}}{B_{i}\!-\!{B_{j}}}\notag\\
   &\!\stackrel{(b)}{=}\!\log N\!+\!\log \rho\!-\!\frac{\kappa}{\ln2}\!+\!\sum_{i=1}^K \log B_i \!\prod_{j=1,j\neq i}^K\frac{B_{i}}{B_{i}\!-\!{B_{j}}},
  \end{align}
  where $(a)$ holds since $\small{\lambda_{E_k}=\rho NB_k\left(1+\frac{d_{SE_k}^{-\alpha}}{NB_k}\right)\rightarrow \rho NB_k}$  as $N\rightarrow\infty$,
   and $(b)$ is based on the fact that $\sum_{i=1}^K\prod_{j=1,j\neq i}^K\frac{B_{i}}{B_{i}\!-\!{B_{j}}}=1$,
  as proved in \cite[Chapter 5]{ross2010proba}.

  \hspace{-0em}Combining \eqref{R_s}, \eqref{R_D_a}, \eqref{R_E_cd} and \eqref{R_E_cd3},  the proof follows.
\end{proof}
\begin{Remark}\label{remark_difference}
  Comparing Corollaries \ref{corollary_scale_N} and \ref{corollary_scale_N2}, we evince  that
only the last terms  for  the asymptotic ESR  are different,
  i.e., $\max_{k\in[1:K]}\log B_k$ and $\sum_{i=1}^K\log B_i\prod_{j=1,j\neq i}^K\frac{B_i}{B_i-B_j}$
   for  non-colluding and colluding Eves, respectively. In addition, the ESRs for
    both non-colluding and colluding Eves
   have the same scaling law, i.e.,  $\log N$.
\end{Remark}
}
%

  \begin{figure}[tbp]
 \begin{center}
{
\includegraphics[width=0.9\columnwidth]{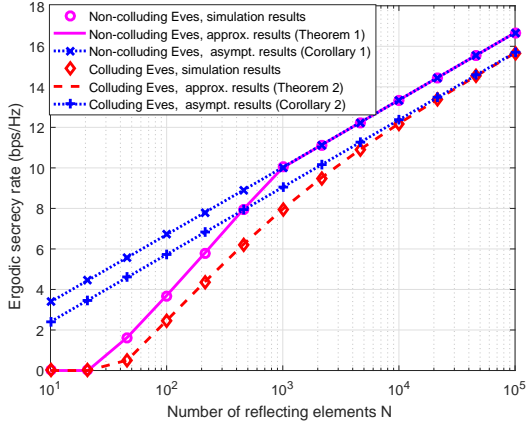}}
\end{center}
\vspace*{-1.5em}
\caption{Ergodic secrecy rate vs. $N$, for $K=5$.}\label{Capacity_N}\vspace{-0.5em}
\end{figure}

  \vspace{-1em}
  \subsection{Large Number of Eves}
  To obtain more insights from the obtained ESR, we provide a simplified
  expression of $R_s$ for large values of $K$ in the following corollary.
  \begin{Corollary}\label{corollary_largeK}
When $N$ and $K\rightarrow \infty$, the ESR for colluding Eves
 can be approximated  as follows
\begin{equation}\small\label{R_s_c3}
R_s\approx \left[R_D-\log\left(1+\sum_{k=1}^K{\lambda_{E_k}}\right)\right]^+.
  \end{equation}
  \end{Corollary}
\begin{proof}
 Based on Lemma \ref{lemma_indep}, $\small{{\rm Var}\left[\sum_{k=1}^K\gamma_{E_k}\right]\rightarrow
 \sum_{k=1}^K\lambda_{E_k}^2}$ as $N\rightarrow \infty$.
 Thus, $\small{\frac{{\rm Var}\left[\sum_{k=1}^K\gamma_k\right]}{\left(\mathbb{E}\left[\sum_{k=1}^K\gamma_k\right]\right)^2}
     \rightarrow \frac{\sum_{k=1}^K\lambda_{E_k}^2}{\left(\sum_{k=1}^K\lambda_{E_k}\right)^2}\rightarrow0}$
     as $N$ and $K\rightarrow \infty$. According to \cite[Theorem 4]{sanayei2007opportun},
    $R_E$ in \eqref{R_E_cd} can be approximated as 
       {\small{$R_E\approx \log\left(1+\mathbb{E}\left[\sum_{k=1}^K\gamma_k\right]\right).$}}
     This completes the proof.
     \end{proof}



\vspace{-1em}
\section{Numerical Results}
In this section, numerical results are provided to verify the analytical results stated in the   theorems
and corollaries. 
 For illustrative purposes,
we set $\alpha=3$, $b=3$ bits, $P=20$ dBm, $\sigma^2=-96$ dBm and $\eta=0.8$. In addition, $S$, $R$ and $D$ are located at
$(0,0)$ m, $(100,0)$ m and $(90,20)$ m, respectively.
As for  the Eves, $E_k$  is located at
$\left(\frac{90k}{K},-20\right)$ m, $k\in[1:K]$.
   For the considered simulation setup, the ESR is zero in the absence of  RIS.
 In this case, in fact, the ESR for non-colluding Eves can be expressed as {\small{$R_{s} =\left[f_{3}(d_{SD})-\max_{k\in[1:K]}f_{3}(d_{SE_k})\right]^+,$}}
 where  $f_3(x)\triangleq \frac{1}{\ln2}\int_{0}^{\infty}{e^{-\frac t{\rho x^{-\alpha}}}}/{(1+t)}{\rm d}t$, $x>0$. Therefore, $R_{s}=0$
     since $f_3(x)$ is a decreasing function of $x$
    and $d_{SD}\geq d_{SE_k}$ in the considered network  configuration, $\forall k\in[1:K]$.
%
%

\begin{figure}[tbp]
 \begin{center}
{
\includegraphics[width=0.9\columnwidth]{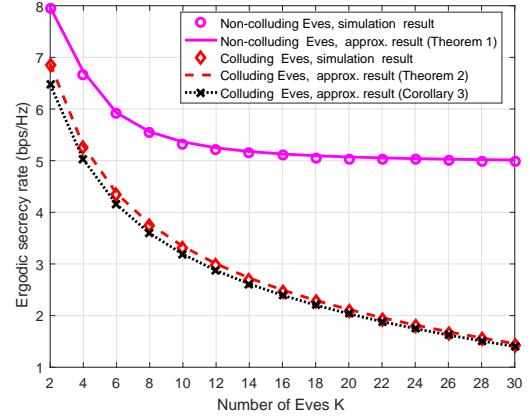}}
\end{center}
\vspace*{-1.5em}
\caption{Ergodic secrecy rate  vs. $K$, for $N=250$.}\label{Capacity_K}
\end{figure}

Fig. \ref{Capacity_N} shows the impact of the number of reconfigurable  elements $N$ on the ESR, when the number of Eves is $K=5$.
We observe  that the approximated analytical results in Theorems \ref{theorem_non} and  \ref{theorem_cd} match well with
 Monte Carlo simulations almost for all values of $N$. This is because   the ESR is zero for small values of $N$,  and starts to  be  positive  for large values of $N$, i.e.,
 $N\geq 46$ in the considered case.
 We also  observe that the ESR increases with $N$.
 For example,  the  ESRs are about $3.7$ bps/Hz and $2.5$ bps/Hz  for non-colluding and colluding Eves,
 respectively, if $N=100$.
 In addition,  the  analytical results obtained in Corollaries \ref{corollary_scale_N} and  \ref{corollary_scale_N2} asymptotically approach the simulations
 as $N$ becomes sufficiently large, which confirms the scaling laws.
The setup with non-colluding
Eves provides a larger secrecy rate  since only the ``best'' Eve  determines  the 
ESR. There exists a constant gap of  about $1$ bps/Hz between the ESRs for  non-colluding and colluding Eves
if $N$ exceeds $10^4$.

In Fig. \ref{Capacity_K}, the  ESR as a function of the number of Eves $K$ for $N=250$ is shown.
The figure confirms the findings in Corollary \ref{corollary_largeK} for colluding Eves, and we observe that the
approximation in \eqref{R_s_c3} becomes
tighter as $K$ increases. 
 The ESR for non-colluding Eves is less affected  by  $K$. We observe, in particular, that there exists an ESR floor of about $5$ bps/Hz for large values of $K$.
  This is because the ESR for non-colluding Eves is determined by the nearest Eve to the source. In
the considered simulation setup, the nearest Eve is located at around $(0,  -20)$ m, when $K$ is large.




\vspace{-1em}
\section{Conclusions}
This letter investigated the ESR of an RIS-assisted communication system in the presence of
discrete phase shifts and
multiple Eves. We obtained an approximated closed-form expression of the
 ESR  
  and unveiled that the ESR   scales with $\log N$
 in the presence of both non-colluding and colluding Eves.
An  interesting   future direction is to analyze the  ESR of  RIS-assisted systems in
the presence of  multi-antenna transmitters.
For example, the recent research works in \cite{IRS_Marco,Marco_TWC} could be generalized in order to take into account security constraints.


\vspace{-1em}
 \appendices
\section{Proof of Lemma \ref{lemma_phase}}\label{proof_lemma_phase}
Let us denote $X_{1,n}\triangleq f_2(\phi_n^*,\theta_{SR,n})$, $X_{2,n}\triangleq\theta_{k,n}$
and $Y_{n}\triangleq \psi_{k,n}$. Accordingly, $Y_{n}=X_{1,n}+X_{2,n}$ as shown in Section \ref{section_model}.
We note that the random phases $X_{1,n}$, $X_{2,n}$ and $Y_{n}$
have a uniform circular
distribution.

Given $X_{1,n}=x_1$, $\forall x_1\in [0,2\pi)$, $Y_{n}=x_1+X_{2,n}$ is uniformly distributed in $[x_1,x_1+2\pi)=[x_1,2\pi)\cup [2\pi,x_1+2\pi)$.
       Since $Y_{n}$ has a circular uniform distribution, $[2\pi,x_1+2\pi)$ is equivalent to
        $[0,x_1)$. Thus, $Y_{n}$ is uniformly distributed
in $[0,2\pi)$, which proves Lemma \ref{lemma_phase}-a).

Since $X_{1,n}$ and $X_{2,n}$ are independent, their joint PDF is
\begin{equation}\small
f_{X_{1,n},X_{2,n}}(x_1,x_2)\!=\!f_{X_{1,n}}(x_1)f_{X_{2,n}}(x_2)\!=\!\frac{1}{2\pi}f_{X_{1,n}}(x_1).
\end{equation}
We can construct the following Jacobian matrix
\begin{equation}\small
  J_{X_{1,n},Y_{n}}(x_1,x_2)=\left[\begin{matrix}
    \frac{\partial x_1}{\partial x_1},&\frac{\partial x_1}{\partial x_2}\\
    \frac{\partial y}{\partial x_1},&\frac{\partial y}{\partial x_2}
  \end{matrix}\right]=\left[\begin{matrix}
    1,&0\\1,&1
  \end{matrix}\right].
\end{equation}
Thus, the joint PDF of $X_{1,n}$ and $Y_{n}$ can be written as
 \begin{small} \begin{eqnarray}
f_{X_{1,n},Y_{n}}(x_1,y)&=&\frac{f_{X_{1,n},X_{2,n}}(x_1,x_2)}{\det(J_{X_{1,n},Y_{n}}(x_1,x_2))}
  =\frac{1}{2\pi}f_{X_{1,n}}(x_1)\notag\\
  &=&f_{X_{1,n}}(x_1)f_{Y_{n}}(y),
\end{eqnarray}\end{small}
\hspace{-0.2em}which implies  that $Y_{n}$ is independent of $X_{1,n}$.
Thus, Lemma \ref{lemma_phase}-b) is proved.



For $\forall i\neq j$ and $i,j\in[1:N]$, we have $Y_{i}=X_{1,i}+X_{2,i}$ and  $Y_{j}=X_{1,j}+X_{2,j}$.
Although the same random phase $\theta_{SD}$ is present  in both $Y_i$ and $Y_j$ as shown in Section \ref{section_model}, $Y_i$ is still   independent of  $Y_j$, due to the following two facts:
(i)
 $Y_{i}$ and $Y_j$ are independent of $X_{1,i}$
and $X_{1,j}$, respectively, according to Lemma \ref{lemma_phase}-b); (ii)
$X_{2,i}$ is independent of $X_{2,j}$.
Therefore, Lemma \ref{lemma_phase}-c) is proved.

\bibliographystyle{IEEEtran}
\bibliography{refference}

\end{document}